\documentclass[preprint]{elsarticle}

\usepackage{array,xspace,multirow,hhline,tikz,colortbl,tabularx,booktabs,fixltx2e,amsmath,amssymb,amsfonts,amsthm}
\usepackage{algorithm}
\usepackage{algorithmic}
\usepackage{verbatim,ifthen}
\usepackage{enumitem}
\usepackage{pifont}
\usepackage{ifthen}
\usepackage{calrsfs,mathrsfs}
\usepackage{bbding,pifont}
\usepackage{pgflibraryshapes}
\usepackage{nicefrac}
\usepackage{url}
\definecolor{light-gray}{gray}{0.9}

	\newtheorem{lemma}{Lemma}%
	\newtheorem{proposition}{Proposition}%

    \newtheorem*{theorem*}{Theorem}
        \newtheorem{definition}{Definition}

		\newcommand{\ml}[1][]{\ifthenelse{\equal{#1}{}}{\mathit{ML}}{\mathit{ML}(#1)}}
		\newcommand{\sml}[1][]{\ifthenelse{\equal{#1}{}}{\mathit{SML}}{\mathit{SML}(#1)}}
		\newcommand{\sd}[1][]{\ifthenelse{\equal{#1}{}}{\mathit{SD}}{\mathit{SD}(#1)}}
		\newcommand{\rsd}[1][]{\ifthenelse{\equal{#1}{}}{\mathit{RSD}}{\mathit{RSD}(#1)}}
		\newcommand{\st}[1][]{\ifthenelse{\equal{#1}{}}{\mathit{ST}}{\mathit{ST}(#1)}}
		\newcommand{\bd}[1][]{\ifthenelse{\equal{#1}{}}{\mathit{BD}}{\mathit{BD}(#1)}}
		\newcommand{\pc}[1][]{\ifthenelse{\equal{#1}{}}{\mathit{PC}}{\mathit{PC}(#1)}}
		\newcommand{\dl}[1][]{\ifthenelse{\equal{#1}{}}{\mathit{DL}}{\mathit{DL}(#1)}}
		\newcommand{\ul}[1][]{\ifthenelse{\equal{#1}{}}{\mathit{UL}}{\mathit{UL}(#1)}}

		\newcommand{\serdict}[1][]{\ifthenelse{\equal{#1}{}}{\sigma}{\sigma(#1)}}

\newcommand{\pav}[0]{\ensuremath{\mathit{PAV}}\xspace}

	\newcommand\eat[1]{}

	\usepackage{enumitem}
	\setenumerate[1]{label=\rm(\it{\roman{*}}\rm),ref=({\it\roman{*}}),leftmargin=*}
	\newlength{\wordlength}

	\newcommand{\midd}{\mathbin{:}}

	\newcommand{\eqclass}[2][]{\ifthenelse{\equal{#1}{}}{[#2]}{[#2]_{\sim_{#1}}}}

	\newcommand{\cross}{\ding{55}}

	\newcommand{\Pref}[1][]{
		\ifthenelse{\equal{#1}{}}{\mathrel R}{\mathop{R_{#1}}}
	}                                          
	\newcommand{\sPref}[1][]{                  
		\ifthenelse{\equal{#1}{}}{\mathrel P}{\mathop{P_{#1}}}
	}                                          
	\newcommand{\Indiff}[1][]{                 
		\ifthenelse{\equal{#1}{}}{\mathrel I}{\mathop{I_{#1}}}
	}
	\newcommand{\prefset}[1][]{\ifthenelse{\equal{#1}{}}{\mathcal{R}}{\mathcal{R}_{#1}}}

\definecolor{PurplePlum}{rgb}{0.1,0,0.55} 
\definecolor{Brown}{rgb}{0.5,.25,0}
\definecolor{Green}{rgb}{0,.5,0}
\definecolor{Orange}{rgb}{1,.5,0}
\definecolor{Gray}{rgb}{0.5,0.5,0.5}
\definecolor{Black}{rgb}{0,0,0}

\usepackage{color} 
\usepackage[normalem]{ulem} 

\newcommand{\jr}{\ensuremath{\mathit{JR}}\xspace}
\newcommand{\ejr}{\ensuremath{\mathit{EJR}}\xspace}

\newcommand{\pjr}{\ensuremath{\mathit{PJR}}\xspace}

\newcommand{\calA}{{\vec{A}}}

\begin{document}


	\title{A Polynomial-time Algorithm to Achieve\\ Extended Justified Representation
	}

	\author{Haris Aziz}
	 \ead{haris.aziz@data61.csiro.au}
      \address{Data61, CSIRO and UNSW Australia\\
  Computer Science and Engineering,
     Building K17, UNSW, Sydney NSW 2052, Australia}

		\author{Shenwei Huang} \ead{shenwei.huang@unsw.edu.au}
	      \address{UNSW Australia\\
	      Computer Science and Engineering,
	      Building K17, UNSW, Sydney NSW 2052, Australia}
	


	\begin{abstract}
				We consider a committee voting setting in which each voter approves of a subset of candidates and based on the approvals, a target number of candidates are to be selected. In particular we focus on the axiomatic property called extended justified representation (EJR). Although a committee satisfying EJR is guaranteed to exist, the computational complexity of finding such a committee has been an open problem and explicitly mentioned in multiple recent papers.
We settle the complexity of finding a committee satisfying EJR by presenting a polynomial-time algorithm for the problem. 
Our algorithmic approach may be useful for constructing other multi-winner voting rules.
	\end{abstract}



\maketitle

\section{Introduction}


The topic of multi-winner/committee voting has witnessed a renaissance with a number of new and interesting developments in the last few years~(see \citep{ABES17a,FSST17a} for recent surveys). We consider a committee voting setting in which each voter approves of a subset of candidates and based on the approvals, a target $k$ number of candidates are selected. The setting has been referred to as approval-based multi-winner voting or committee voting with approvals. The setting has inspired a number of natural voting rules~\citep{Kilg10a,BrFi07c,LMM07a,AGG+15a,SFL16a}. Many of the voting rules attempt to satisfy some notion of representation. However it has been far from clear what axiom captures the representation requirements.

\citet{ABC+15a,ABC+16a} proposed two compelling representation axioms called \emph{justified representation (\jr)} and \emph{extended justified representation (\ejr)}. 
Interestingly, \citet{SFFB16a} presented an intermediate property called \emph{proportional justified representation (\pjr)}. \footnote{The property \pjr was independently proposed by Haris Aziz in October 2014 who referred to it as weak \ejr.} 
The idea behind all the three properties is that a cohesive and large enough group deserves sufficient number of approved candidates in the winning set of candidates. Interestingly, it is known that there always exists a committee satisfying the strongest property~\ejr~\citep{ABC+16a}. However to date, it has been unknown whether a committee satisfying \ejr can be computed in polynomial time. For the two weaker representation notions, 
polynomial-time algorithms have been presented for finding a committee satisfying \jr~\citep{ABC+15a,ABC+16a}\footnote{For \jr, a simple linear-time algorithm call GreedyAV finds a committee satisfying \jr.} and \pjr~\citep{BFJL16a,SFF16a}\footnote{It has recently been shown that a committee satisfying \pjr can be computed in polynomial time. \citet{BFJL16a} proved that SeqPhragm\'{e}n (an algorithm proposed by Swedish mathematician Phragm\'{e}n in the 19th century) is polynomial-time and returns a committee satisfying \pjr.  Independently and around the same time as the result by \citet{BFJL16a}, \citet{SFF16a} presented a different algorithm that finds a \pjr committee and also satisfies other desirable monotonicity axioms.}.
On the other hand, the computational complexity of finding a committee satisfying \ejr has been open.  \citet{ABC+15a,ABC+16a} mentioned the problem in their original paper. The problem has been reiterated in subsequent work. \citet{BFJL16a} state that 
\begin{quote}
``\emph{it remains an open problem whether committees providing EJR can be computed efficiently}.''
\end{quote}
\citet{SFF16a} mention the same problem:
\begin{quote}
``\emph{Whether a voting rule exists that satisfies the extended justified representation and can be computed in polynomial time remains an open issue.}'' 
\end{quote}

In a different paper, \citet{SFF+17a} state the following.
\begin{quote}
\emph{``In contrast, it is conjectured that finding committees that provide EJR is computationally hard.''}
\end{quote}

Incidentally, there exists an interesting rule called \pav (\emph{Proportional Approval Voting}) that satisfies \ejr~\citep{ABC+16a}. In \pav, each voter is viewed as getting an additional score of $1/j$ for getting the $j$-th approved candidate in the committee. The \pav rule returns a committee with the highest total \pav score for the voters.
The \pav rule has a fascinating history as it was  proposed by the Danish polymath Thorvald N. Thiele in the 19th century and then rediscovered by Forrest Simmons~\citep{Jans16a}. Finding a \pav outcome is NP-hard~\cite{AGG+15a,SFL16a} and W[1]-hard even if each voter approves of 2 candidates~\cite{AGG+15a}. 
Thiele also presented a greedy sequential version of PAV. The rule that is referred to as SeqPAV (Sequential PAV) or RAV (reweighted approval voting) does not even satisfy \jr~\citep{ABC+16a}.



One natural approach to find a committee satisfying \ejr is to enumerate possible committees and then test them for \ejr. However the number of committees is exponential and even testing whether a committee satisfies \ejr is coNP-complete \cite{ABC+16a}. \citet{ABC+16a} presented a result that implies that if $k$ is a constant, then a committee satisfying \ejr can be computed in time $poly(n\cdot |C|^k)$. The result does not show that finding a committee satisfying \ejr is polynomial-time solvable in general or whether it is fixed parametrized tractable.

\bigskip
\paragraph{Contributions}

We present the first polynomial-time algorithm to find a committee that satisfies \ejr. The result implies that there exists a polynomial-time algorithm to find a committee that satisfies the weaker property of \pjr. As mentioned earlier, it has only recently been proven in two independent papers that a committee satisfying \pjr can be computed in polynomial time~\citep{BFJL16a,SFF16a}. Both of the algorithms in~\citep{BFJL16a} and \citep{SFF16a} sequentially build a committee while optimizing some flow or load balancing objective. 
In contrast, our algorithm uses an approach based on swapping candidates from inside a committee with candidates from outside the committee. 
The correctness of our algorithm relies on a careful insight on the connection between \ejr and a property we refer to  as  \pav-swap-freeness. 
We feel that this simple idea of allowing swaps may lead to other interesting algorithms for \ejr as well as other compelling properties in multi-winner voting problems.





\section{Approval-based Committee Voting and Representation Properties}

We consider a social choice setting with a set $N=\{1,\ldots, n\}$ of voters and a set $C$ of $m$ candidates. 
Each voter $i\in N$ submits an approval ballot $A_i\subseteq C$, which represents the subset of candidates that she
approves of. We refer to the list $\calA = (A_1,\ldots, A_n)$ of approval ballots as the {\em ballot profile}. 
We will consider {\em approval-based multi-winner voting rules} that take as input a quadruple $(N, C, \calA, k)$, 
where $k$ is a positive integer that satisfies $k\le m$, and return a subset 
$W \subseteq C$ of size $k$, which we call the {\em winning set}, or {\em committee}. 

%


\begin{definition}[Justified representation (JR)]
Given a ballot profile $\calA = (A_1, \dots, A_n)$ over a candidate set $C$ and a target committee size $k$,
we say that a set of candidates $W$ of size $|W|=k$ {\em satisfies justified representation 
for $(\calA, k)$} if 
\[\forall X\subseteq N: |X|\geq \frac{n}{k} \text{ and } |\cap_{i\in X}A_i|\geq 1 \implies (|W\cap (\cup_{i\in X}A_i)|\geq 1)\]
\end{definition}

The rationale behind this definition is that if $k$ candidates are to be selected, then, intuitively,
each group of $\frac{n}{k}$ voters ``deserves'' a representative. Therefore, a set of $\frac{n}{k}$ voters 
that have at least one candidate in common should not be completely unrepresented.

\begin{definition}[Proportional Justified Representation (\pjr)]
Given a ballot profile $(A_1, \dots, A_n)$ over a candidate set $C$, a target committee size $k$, $k\le m$, and integer $\ell$
we say that a set of candidates $W$, $|W|=k$, {\em satisfies $\ell$-proportional justified representation
for $(\calA, k)$}  if
\[\forall X\subseteq N: |X|\geq \ell\frac{n}{k} \text{ and } |\cap_{i\in X}A_i|\geq \ell \implies (|W\cap (\cup_{i\in X}A_i)|\geq \ell)\]

We say that $W$ {\em satisfies proportional justified representation for $(\calA, k)$} if it {satisfies $\ell$-proportional justified representation
for $(\calA, k)$} and all integers $\ell\leq k$.
\end{definition}

\begin{definition}[Extended justified representation (\ejr)]
Given a ballot profile $(A_1, \dots, A_n)$ over a candidate set $C$, a target committee size $k$, $k\le m$,
we say that a set of candidates $W$, $|W|=k$, {\em satisfies  $\ell$-extended justified representation
for $(\calA, k)$}  and integer $\ell$ if
\[\forall X\subseteq N: |X|\geq \ell\frac{n}{k} \text{ and } |\cap_{i\in X}A_i|\geq \ell \implies (\exists i\in X: |W\cap A_i|\geq \ell).\]

We say that $W$ {\em satisfies extended justified representation for $(\calA, k)$} if it {satisfies $\ell$-extended justified representation
for $(\calA, k)$} and all integers $\ell\leq k$.

\end{definition}

It is easy to observe that \ejr implies \pjr which implies \jr. So any committee that satisfies \ejr also satisfies the other two properties.

%
%

\section{PAV-score and Swaps}


The \pav-score of a voter $i$ for a committee $W$ is 
\[H(|W\cap A_i|)\]
where 
\[
H(p)= \begin{cases}
0, \text{ for $p=0$}\\
\sum_{j=1}^p\frac{1}{j}, \text{ for $p>0$.}\\
\end{cases}
\]

The \pav-score of a committee $W \subseteq C$ 
is defined as 
\[\sum_{i\in N}H(|W\cap A_i|).\] 

The $\pav$ rule that we discussed in the introduction outputs a set $W \subseteq C$ of size $k$
with the highest \pav-score. 




We say that a committee $W$ such that $|W|=k$ satisfies \emph{\pav-swap-freeness} if there exists no $c'\in W, c\in C\setminus W$ s.t. $\textit{PAV-score}((W\setminus \{c'\})\cup  \{c\})>\textit{PAV-score}((W))$.
Note that if a committee $W$ has the highest possible \pav-score, it satisfies \pav-swap-freeness.

We now define marginal contribution as used in \citep{ABC+16a}. 
           For each candidate $w\in W$, we define $MC(w,W)$ its {\em marginal contribution} as the difference between the \pav-score of $W$ and that of $W\setminus\{w\}$:
	\[MC(w,W)=\textit{PAV-score}((W)-\textit{PAV-score}((W\setminus \{w\}).\]
	
Let $MC(W)$ denote the sum of marginal contributions of all candidates in $W$:
\[MC(W)=\sum_{w\in W}MC(w,W).\]

We now formally state as a lemma an observation that was already made in \citep{ABC+16a}.

	%

	\begin{lemma}\label{lemma:mc}
		For any committee $W$ such that $|W|=k$,
		$\sum_{c\in W}MC(c,W)\leq |\{i\in N\midd A_i\cap W\neq \emptyset\}|$. Moreover there exists at least one $c\in W$ such that $MC(c,W)\leq |\{i\in N\midd A_i\cap W\neq \emptyset\}|/k\leq n/k$.
		\end{lemma}
		\begin{proof}

			Pick a voter $i\in N$, and let $j=|A_i\cap W|$. If $j>0$,  this voter contributes exactly $\frac{1}{j}$ to the marginal contribution of each candidate in $A_i\cap W$, 
			and hence her contribution to $MC(W)$ is exactly $1$.
			If $j=0$, this voter does not contribute to $MC(W)$ at all.
			Therefore, we have $MC(W)=\sum_{c\in W}MC(c,W)\leq |\{i\in N\midd A_i\cap W\neq \emptyset\}|\leq n$. 
			Since there are exactly $k$ candidates, there exists some $c\in W$ such that $MC(c,W)\leq |\{i\in N\midd A_i\cap W\neq \emptyset\}|/k\leq n/k$.
			\end{proof}

	We now prove that if a committee satisfies \pav-swap-freeness, then it satisfies \ejr. The argument is almost identical to the argument that the outcome of \pav satisfies \ejr\citep[][]{ABC+16a}. However, we reproduce it just for the sake of completeness because we will further refine this argument.

		\begin{lemma}\label{thm:pav-ejr}
		If a committee satisfies \pav-swap-freeness, then it satisfies \ejr.
		\end{lemma}
		\begin{proof}
			Suppose that there is a committee $W$ such that $|W|=k$ that satisfies \pav-swap-freeness but violates \ejr. 
		Since $W$ violates \ejr, there is a value of $\ell\geq 1$ and a set of voters $N^*$,
		$|N^*| = s \ge \ell\cdot \frac{n}{k}$. 
		We know that at least one of the $\ell$ candidates approved by all voters in $N^*$ is not elected; let $c$ be some such candidate.
	Each voter in $N^*$ has at most $\ell-1$ representatives in $W$, so the marginal contribution of $c$ (if it were to be added to $W$) would be at least $s\cdot \frac{1}{\ell} \ge \frac{n}{k}$. On the other hand, by Lemma~\ref{lemma:mc}, we have $\sum_{c\in W}MC(c,W)\leq n$.

		Now, consider some candidate $w\in W$ with the smallest marginal
		contribution; clearly, his marginal contribution is at most $\frac{n}{k}$. 
		If it is strictly less than $\frac{n}{k}$, we are done, as we can improve the
		total \pav-score by swapping $w$ and $c$, a contradiction.

		 Therefore
		suppose it is exactly $\frac{n}{k}$, and therefore the marginal contribution
		of each candidate in $W$ is exactly $\frac{n}{k}$. We know that 
	$A_i\cap W\neq \emptyset$ for each $i\in N^*$, because otherwise $\sum_{w'\in W}MC(w',W)\leq n-1$ (by Lemma~\ref{lemma:mc}) which implies that the marginal contribution of $w$ is less than $\frac{n}{k}$.
Hence we know that $A_i\cap W\neq\emptyset$ for all $i\in N^*$. Pick some
		candidate $w'\in W\cap A_i$ for some $i\in N^*$, and set $W'=(W\setminus\{w'\})\cup \{c\}$.
		Observe that after $w'$ is removed, adding $c$ increases
		the total \pav-score by at least 
	\begin{align*}
	(s-1)\cdot \frac{1}{\ell}+\frac{1}{\ell-1}
	>\frac{s}{\ell}
	\geq n/k.
	\end{align*}

		Thus, the \pav-score of $W'$ is higher than that of $W$, a contradiction again.
	\end{proof}

Although \pav-swap-freeness is a much weaker property than maximizing total \pav-score,  it is surprising that it already implies \ejr. In the next section, this insight helps us to come up with useful algorithms.

		\section{MaxSwapPAV}
		
Based on \pav-score improving swaps, one can formulate the following algorithm called SwapPAV.
		\begin{quote}
			SwapPAV: Start from a random committee of size $k$. Keep implementing swaps that increase the total \pav-score of the committee while such a swap is possible. Return the committee if no more improving swaps are possible. 
			\end{quote}
		
Our first observation is that Swap-PAV always terminates. The reason is that each time we implement the swap, the PAV-score of the committee increases. This can only happen finitely often as $\textit{PAV-score}((W)\le nH(k)\le n(\ln k +1)$ for any committee $W$ of size $k$.
In fact, one can easily prove that with each improving swap, the \pav-score increases by at least $1/k!$ so that that the total number of swaps cannot exceed $n(\ln k+1)k!$. This observation already gives us the first FPT algorithm for finding a committee satisfying \ejr. 

		We now show how we can modify SwapPAV to find a committee satisfying \ejr in polynomial time. We  modify SwapPAV as follows. If $W$ is not \pav-swap-free, then we look at all possible swaps and only implement the swap which makes biggest difference to the \pav-score. 
We impose an extra condition that we only swap if the improvement in the total \pav-score is at least $\frac{1}{2k^3}$.		
The algorithm is specified as Algorithm~\ref{algo:maxswappav} (MaxSwapPAV).

							\begin{algorithm}[h!]
								  \caption{MaxSwapPAV}
								  \label{algo:maxswappav}

								\begin{algorithmic}
									\REQUIRE  $(N,\calA, k)$.
									\ENSURE $W$
								\end{algorithmic}
								\begin{algorithmic}[1]
									\STATE $W\longleftarrow $ any committee of size $k$.
									
			%
			%
			%
			%
			%
									
															\WHILE{$\exists c'\in W, c\in C\setminus W$ s.t. $\textit{PAV-score}((W\setminus \{c'\})\cup  \{c\})-\textit{PAV-score}(W)\geq 1/2k^3$}
															\FOR{each $c'\in W, c\in C\setminus W$}
\STATE $\textit{diff}(c,c')=\textit{PAV-score}((W\setminus \{c'\})\cup  \{c\})-\textit{PAV-score}(W)$ 				
\ENDFOR
\STATE Find $c'\in W, c\in C\setminus W$ with the maximum $\textit{diff}(c,c')$.
									\STATE $W\longleftarrow (W\setminus \{c'\})\cup  \{c\}$
									\ENDWHILE

								
									\RETURN $W$.
								\end{algorithmic}
							\end{algorithm}

We now argue why MaxSwapPAV returns a committee satisfying \ejr and it terminates in polynomial time. The most crucial lemma for both statements is Lemma~\ref{lemma:crucial}. The lemma is stronger than Lemma~\ref{thm:pav-ejr} and requires a more careful analysis.


	\begin{lemma}\label{lemma:crucial}
	Suppose that $W$ does not satisfy \ejr, then there exist $c'\in W, c\in C\setminus W$ s.t. $\textit{PAV-score}((W\setminus \{c'\})\cup  \{c\})-\textit{PAV-score}((W))\geq 1/2k^3$.	
	\end{lemma}
	
	\begin{proof}
	Since $W$ violates \ejr, there is a value of $\ell\ge 1$ and a set of voters $N^*$,
	$|N^*| = s \ge \ell\cdot \frac{n}{k}$. 
	We know that at least one of the $\ell$ candidates approved by all voters in $N^*$ is not elected; let $c$ be some such candidate.
	Each voter in $N^*$ has at most $\ell-1$ representatives in $W$, so the marginal contribution of $c$ (if it were to be added to $W$) would be at least $s\cdot \frac{1}{\ell} \ge \frac{n}{k}$.
	
	Let $w\in W$ be the candidate with smallest marginal contribution. By Lemma \ref{lemma:mc}, $MC(w,W)\le \frac{n}{k}$.
	If $MC(w,W)\le \frac{n}{k}-\frac{1}{2k^3}$, then replacing $w$ with $c$ results in a committee $W'$ with PAV-score
	increasing by at least $\frac{1}{2k^3}$, and so we are done.
	
	So, we assume that $MC(w,W)> \frac{n}{k}-\frac{1}{2k^3}$.
This implies that $MC(w',W)> \frac{n}{k}-\frac{1}{2k^3}$
	for every $w'\in W$. Since $\sum_{w'\in W}MC(w',W)\le n$, we use this fact to find the maximum possible marginal contribution among all candidates in $W$. Let the maximum marginal contribution be $MC(b,W)$ of candidate $b$. In that case we know that
	\begin{align*}
		&\sum_{w'\in W}MC(w',W)\leq n\\
		\iff&\sum_{w'\in W\setminus \{b\}}MC(w',W)+MC(b,W)\leq n\\
		\iff&MC(b,W)\leq n-\sum_{w'\in W\setminus \{b\}}MC(w',W)\\
		\implies&MC(b,W)< n-(k-1)(\frac{n}{k}-\frac{1}{2k^3})\\
				\iff&MC(b,W)< \frac{(2nk^3-2nk^3+2nk^2+k-1)}{2k^3}\\
				\iff&MC(b,W)< \frac{n}{k}+\frac{1}{2k^2}-\frac{1}{2k^3}<\frac{n}{k}+\frac{1}{2k^2}.
	\end{align*}

	 Hence, it follows that $MC(w',W)<\frac{n}{k}+\frac{1}{2k^2}$ for every
	$w'\in W$.

	We now claim that there is a candidate $w'\in W$ that is also in $\bigcup_{i\in N^*}A_i$.
	Suppose not. This means that no one in $N^*$ approves of anybody in $W$ and 
	so by Lemma~\ref{lemma:mc}, $\sum_{w'\in W}MC(w',W)\le |N\setminus N^*|\le n-\frac{n}{k}$.
	Thus, $MC(w,W)\le \frac{n}{k}-\frac{n}{k^2}<\frac{n}{k}-\frac{1}{2k^3}$, contradicting our assumption that
	 $MC(w,W)> \frac{n}{k}-\frac{1}{2k^3}$. Now pick any $w'\in W\cap \bigcup_{i\in N^*}A_i$.
	 As $w\in \bigcup_{i\in N^*}A_i$, this implies that $|A_i\cap (W\setminus \{w'\})|\leq \ell-2$ for some $i\in N^*$. Hence, 
	 $MC(c,(W\setminus \{w'\})\cup \{c\})$ would be at least $\frac{n}{k}+\frac{1}{\ell-1}-\frac{1}{\ell}\ge \frac{n}{k}+\frac{1}{k^2}$.
	 Therefore, replacing $w'$ with $c$ results in a committee $W'$ with PAV-score increasing by at least 
	 $\frac{1}{k^2}-\frac{1}{2k^2}=\frac{1}{2k^2}\ge \frac{1}{2k^3}$. This proves the lemma.
	\end{proof}
	
Lemma~\ref{lemma:crucial} is the foundation for proving the main properties of the MaxSwapPAV algorithm.
	
	\begin{proposition}
					MaxSwapPAV returns a committee that satisfies \ejr. 
					\end{proposition}
					\begin{proof}
					MaxSwapPAV returns a committee $W$ such that there exist no $c'\in W$ and $c\in C\setminus W$ s.t. $\textit{PAV-score}((W\setminus \{c'\})\cup  \{c\})-\textit{PAV-score}((W))\geq 1/2k^3$. By Lemma~\ref{lemma:crucial}, such a committee satisfies \ejr.
	\end{proof}
	
	\begin{proposition}
					MaxSwapPAV runs in polynomial time $O(n^2mk^4\ln k)$.
					\end{proposition}
					\begin{proof}
						We first show that the total number of swaps in MaxSwapPAV cannot exceed $2n(\ln k+1)k^3$. In Lemma~\ref{lemma:crucial}, we proved that each swap in the algorithm improves the \pav-score by at least $\frac{1}{2k^3}$. Since the \pav score of any committee cannot exceed $n(\ln k+1)$, there can be at most $2n(\ln k+1)k^3$ swaps. Each swap requires examining $O(km)$ pairs of candidates. For each pair, we need to make $O(n)$ operations. 
	\end{proof}

It follows from the two propositions above  that MaxSwapPAV returns a committee satisfying \ejr in polynomial time.

%
%
%
%
%
%
%

\section{Discussion}

To conclude, we presented the first polynomial-time algorithm for finding a committee that satisfies \ejr. In Table~\ref{table:jrrules}, we summarize the justified representation related properties satisfied by different polynomial-time algorithms in the literature.

\begin{table}[h!]
	
\begin{center}
\begin{tabular}{lccc}
\toprule
&\jr&\pjr&\ejr\\
Rules&&&\\
\midrule
MaxSwapPAV (this paper)&\checkmark&\checkmark&\checkmark\\
SeqPhragm\'{e}n~\citep{BFJL16a,Phra94a}&\checkmark&\checkmark&\cross\\
Open D’Hondt (ODH)~\citep{SFF16a} &\checkmark&\checkmark&\cross\\
GreedyAV~\citep{ABC+16a,Thie95a} &\checkmark&\cross&\cross\\
SeqPAV~\citep{Thie95a} &\cross&\cross&\cross\\
\bottomrule
\end{tabular}
\end{center}
\caption{Related and known polynomial-time algorithms for approval-based committee voting. 
}
\label{table:jrrules}
\end{table}

Our result shows that \ejr is as amenable to efficient computation as \pjr. 
Depending on particular specifications, our algorithm to find a committee satisfying \ejr can also used to formulate particular voting rules.

\citet{SFL16a} mentioned that SeqPAV can be seen as a desirable approximation algorithm for PAV. Our alternative approach of allowing exchanges rather than sequentially building a committee seems to be closer to one of the defining features of \pav that it satisfies \ejr. 

The approach of implementing swaps of candidates also makes it possible to move towards fairer representation from a default committee without having to disband the whole committee. The swapping procedure can also be used as post-processing step after running any other committee rule. If the initial committee is the outcome of SeqPAV, then we know that the committee already guarantees at least $(1-\frac{1}{e})$ of the maximum possible 
\pav-score~\citep{SFL16a}. Hence it follows that subsequent \pav-score improving swaps can only further increase the score.

Our algorithmic result also adds a new talking point to the debate between the harmonic scoring approach of Thiele versus the load balancing approach of Phragm\'{e}n that started over a hundred years ago~\citep{Jans16a}. Note that \citet{BFJL16a} showed that SeqPhragm\'{e}n---one of the efficient algorithms within Phragm\'{e}n's framework of multi-winner rules---satisfies \pjr. On the other hand, SeqPAV the well-known polynomial-time algorithm using  Thiele's approach of harmonic weights does not even satisfy \jr. However, we have shown that by allowing swaps of candidates, one can satisfy \ejr which is stronger than \pjr.


 \end{document}